\documentclass[10pt,conference]{IEEEtran}

\usepackage{mathrsfs}
\usepackage{subfigure}
\usepackage{amsfonts}
\usepackage{amsmath}
\usepackage{graphicx}
\usepackage{subfigure}

\newtheorem{theorem}{Theorem}

\newtheorem{claim}{Claim}

\newtheorem{corollary}[theorem]{Corollary}

\newcommand{\negspace}{\hspace*{-2mm}}

\newcommand{\oE}{\overline{E}}
\newcommand{\oN}{\overline{N}}

\newcommand{\bx}{{\bf x}}
\newcommand{\bY}{{\bf Y}}

\newcommand{\suppress}[1]{}

\newcommand{\cD}{\mathcal{D}}

\newcommand{\cA}{\mathcal{A}}
\newcommand{\cN}{\mathcal{N}}

\title{A Simple Necessary and Sufficient Condition for the Double Unicast Problem}

\author{
\authorblockN{Sagar Shenvi and Bikash Kumar Dey}
\authorblockA{Department of Electrical Engineering \negspace \\
Indian Institute of Technology Bombay\negspace \\
Mumbai, India, 400 076\\
\{sagars,bikash\}@ee.iitb.ac.in}
}

\begin{document}

\maketitle
\begin{abstract}
We consider a directed acyclic network where there are two source-terminal
pairs and the terminals need to receive the symbols generated
at the respective sources. Each source independently generates
an i.i.d. random process over the same alphabet.
Each edge in the network is error-free,
delay-free, and can carry one symbol from the alphabet per use.
We give a simple necessary and sufficient condition for
being able to simultaneously satisfy the unicast requirements of
the two source-terminal pairs at rate-pair $(1,1)$ using vector network
coding. The condition is also
sufficient for doing this using only ``XOR'' network coding and is
much simpler compared to the necessary and sufficient conditions
known from previous work. Our condition also yields a simple characterization
of the capacity region of a double-unicast network which does not support
the rate-pair $(1,1)$.
\end{abstract}

\section{Introduction}\label{section:Introduction}
The field of network coding started with the seminal work of Ahlswede et. al.~\cite{ahlswede1} which showed that throughput improvements could be obtained by allowing intermediate nodes in a network to perform coding operations on incoming packets. The capacity of a multicast network was shown to be equal to the minimum of the min-cuts between the source and the individual terminals. In~\cite{li1} it was shown that linear network coding is sufficient to achieve the multicast capacity. An algebraic formulation was given for the multicast scenario in~\cite{koetter2}. The area of network coding for multicast has since seen rapid developments (e.g.~\cite{jaggi1}~\cite{ho1}).

Though the multicast case is well understood, the promise of network
coding for arbitrary source-terminal demands has not yet been gauged.
The problem of finding necessary and sufficient conditions for an arbitrary 
network information flow problem seems a difficult task~\cite{lehman1}. 
Even for the case of multiple unicast sessions, explicit necessary and
sufficient 
conditions for feasibility of network coding are not known.
An outer bound, called network sharing bound, on the achievable rate-region for arbitrary multi-source
multi-terminal network was presented in \cite{yan1}. This, as a special
case, gives a necessary condition for any rate-pair to be achievable
in a double-unicast network.
Recently, Wang and Shroff~\cite{wang1} gave 
a path-based necessary and sufficient condition for a directed acyclic network to 
be able to support two unicast sessions by scalar network coding. Their condition
is in terms of existence of a pair of triplets of paths satisfying
a certain condition. They also stated~\cite[Corollary 3]{wang1} that the condition implied by the
network sharing bound is also sufficient for this case. 

In this paper, we present a considerably simpler necessary and sufficient
condition for a directed acyclic network to simultaneously support two unicast sessions
at rate $(1,1)$. This characterization does not assume scalar network coding.
We provide a rigorous proof for the condition. Though on analysing for
our case, one can
conclude the equivalence of our condition with that implied by
network sharing bound, our condition is useful for three reasons:
(i) the simplicity of our condition is useful and illuminating, (ii) 
the condition is proved to be necessary and sufficient even under vector network coding, (iii) the condition allows for exact capacity region characterization,
namely $r_1+r_2 \leq 1$,
of networks which do not support the rate $(1,1)$. 
Our proof is constructive, and it also reveals that whenever a network
supports two simultaneous unicast sessions at rate $(1,1)$, a XOR coding scheme suffices.
This was also noted in \cite{wang2}.

The paper is organized as follows. In Section \ref{section:Terminology}, we introduce
some notations and terminology. In Section \ref{section:Results}, we present our
results without proofs. The proof of the main result is
presented in Section \ref{section:proof1}. The paper is concluded
in Section \ref{section:conclusion}.

\section{Notations and Terminology}\label{section:Terminology}
We consider a directed acyclic network represented by a graph $G(V,E)$.
In network coding, the intermediate nodes combine the incoming messages
to construct the outgoing messages. In this paper, we assume that each
source generates a i.i.d. random process over an alphabet $\cA$, and
the processes generated by different sources are independent.
Each edge can carry one symbol from the alphabet per unit time
in a error-free and delay-free manner.
If the alphabet is a field, a network code is said to be linear if
the nodes perform linear combining over that field. In a general
network, there may be multiple sources and multiple terminals demanding
some of the source processes. Any network code which satisfies
the terminals' demands is said to {\it solve} the network, and the
network code is called a {\it solution}.

If a network is solvable by linear coding over the binary field
$F_2$, then clearly it is also solvable over any other field
by nodes performing only addition or subtraction.
Moreover,
for any alphabet, one can define an abelian group structure
isomorphic to the integers modulo $|\cA|$, and if a network 
has a linear coding solution over $F_2$ then it also has
a solution over $\cA$ where each node performs addition and
subtraction in the corresponding abelian group.
We call such a network code as XOR code.

The starting node of an edge $e$ is called its tail, and is denoted
by $tail (e)$. The end-node of $e$ is called its head, and is denoted
by $head (e)$.
A path $P$ from $v_1$ to $v_l$ - also called a $(v_1,v_l)$ path - is a sequence
of nodes $v_1,v_2,\ldots, v_l$
and edges $e_1,e_2,\ldots, e_{l-1}$ such that $v_i = tail (e_i)$
and $head (e_i) = v_{i+1}$ for $1\leq i \leq l-1$.
The node pair $(v_i,v_j)$ is said to be connected if there exists a
$(v_i,v_j)$ path. If there is a $(v_i,v_j)$ path, then we call $v_i$ an
ancestor of $v_j$, and $v_j$ a descendant of $v_i$. A section of the
path $P$ starting from the node $v_j$ and ending at $v_k$
is denoted by $P(v_j:v_k)$. For any path or a section of a path $P$,
$V(P)$ denotes the set of vertices on it.
If $P_1$ is a path from $v_i$ to $v_j$ and $P_2$ is a path from
$v_j$ to $v_k$ then $P_1P_2$ denotes the path from $v_i$ to $v_k$
obtained by concatenating $P_1$ and $P_2$.

In this paper, we consider a directed acyclic double-unicast network
where there are two sources $s_1$ and $s_2$, and two terminals $t_1$ and
$t_2$. The terminal $t_1$ wants to recover the source process
generated by 
$s_1$ and the terminal $t_2$ wants to recover the source process
generated by $s_2$. 
A rate-pair $(r_1,r_2)$ is said to be achievable if for some $n$,
there is an $n$-length vector network code (which uses each link $n$
times) which simultaneously allows communication of a symbol-vector
$\bx_1\in \cA^{r_1n}$ from $s_1$ to $t_1$ and a symbol-vector
$\bx_2\in \cA^{r_2n}$ from $s_2$ to $t_2$.
The capacity region
of a double-unicast network is defined as the set of all achievable rate pairs $(r_1,r_2)$ at which
the two unicasts can be supported and its symmetric capacity $r$ is defined as the maximum $r$ s.t.
$(r,r)$ is achievable.

For the rest of the paper, we assume that the network under consideration
is a directed acyclic double-unicast network where both the source-terminal
pairs $(s_1,t_1)$ and $(s_2,t_2)$ are connected i.e. there exists a path from $s_1$ to $t_1$ and from $s_2$ to $t_2$. Otherwise the capacity-region of the
network is characterised by the max-flow min-cut theorem.

\section{ Main Results}\label{section:Results}
We now present the main result of the paper in the following theorem.
We emphasize that in all that follows, $(s_1,t_1)$ and $(s_2,t_2)$ are
assumed to be connected.

\begin{theorem}\label{theorem:mainres}
 A directed acyclic double-unicast network
can not simultaneously support two unicast sessions $(s_1,t_1)$ and $(s_2,t_2)$
 at rate $(1,1)$ if and only if it contains an edge whose removal disconnects
$(s_1,t_1)$, $(s_2,t_2)$, and at least one of $(s_1,t_2)$ and $(s_2,t_1)$.
Furthermore, (i) whenever such an edge does not exist, the rate-pair $(1,1)$
is achieved by using XOR coding and (ii) whenever such an edge exists,
the capacity region of the network is given by $r_1+r_2 \leq 1$.
\end{theorem}

\begin{corollary}
The symmetric capacity of a double-unicast network is either
$0,1/2,$ or $\geq 1$.
\end{corollary}

\section{Proof of Theorem \ref{theorem:mainres}}\label{section:proof1}

\subsection{Proof of necessity}
We consider two symbols $x_1$ and
$x_2$ generated at $s_1$ and $s_2$ respectively.
In this part of the proof, we show that whenever an edge satisfying the
hypothesis of the theorem does not exist, $x_1$ and $x_2$ can be communicated
to $t_1$ and $t_2$ respectively by XOR coding. 
We consider two cases:

{\bf Case I:} {\it There exist two edge-disjoint $(s_1,t_1)$ and $(s_2,t_2)$ paths}

In this case, clearly the hypothesis of the theorem is satisfied and the 
network is solvable by routing.

{\bf Case II:} {\it There are no two edge-disjoint $(s_1,t_1)$ and
$(s_2,t_2)$ paths}

So for any $(s_1,t_1)$ path,
removing all its edges disconnects $(s_2,t_2)$. Furthermore let us
assume that no edge satisfies the hypothesis of Theorem~\ref{theorem:mainres} i.e., removal of no edge disconnects $(s_1,t_1)$, $(s_2,t_2)$
and one or more of $(s_1,t_2)$ and $(s_2,t_1)$.
We will show that the network is solvable using a XOR code in this scenario.

Let $P$ be a $(s_1,t_1)$ path. Let $n_1$ be the first node on $P$
such that there exists a $(s_2,n_1)$ path and let
$n_k$ be the last node on $P$ such that there exists a $(n_k,t_2)$ path.
WLOG, let us assume that:

1. No ancestor of $n_1$ has paths from both
 the sources to itself (else we could have taken $P$
to be the path which contained this ancestor).

2. No descendant of $n_k$ has paths from itself to both the terminals
(else we could have taken $P$ to be a path
which contained this descendant).

Let us choose a $(s_2,n_1)$ path and a $(n_k,t_2)$ path and call them
respectively $T_1$ and $T_2$ (see Fig. \ref{fig:caseIIA}).

 We consider two cases:

{\bf Case IIA:} {\it There exists an edge $l$ whose removal 
disconnects $(s_1,t_1)$ and $(s_2,t_2)$ but not $(s_1,t_2)$ or
 $(s_2,t_1)$.}

Let $\cD$ be the set of all such edges.
Note that all the edges in $\cD$ are on $P(n_1:n_k)$. 
In fact, it is evident that
all the edges in $\cD$ lie on any path from $n_1$ to $n_k$.

Let $I_{n_1,n_k}$ be the set of intermediate nodes between $n_1$ and $n_k$
including themselves, i.e., $n_1,n_k,$ and the nodes which are both descendants of $n_1$ 
and ancestors of $n_k$.

\begin{claim}\label{claim:allpaths}
(i) Any $(s_1,t_2)$ path that contains any node from
$I_{n_1,n_k}$, contains all edges from $\cD$.
(ii) Any $(s_2,t_1)$ path that contains any node from
$I_{n_1,n_k}$, contains all edges from $\cD$.
\end{claim}
\begin{proof} 
We prove only the first part by contradiction. The second part follows
similarly.
Let $U$ be a $(s_1,t_2)$ path containing
$z\in I_{n_1,n_k}$ and not containing $l\in \cD$ (see Fig.~\ref{fig:caseIIA-a}).
Note that there are paths
from $z$ to both $t_1$ and $t_2$ since $n_k$ is a descendant of $z$.
Without loss of generality, we assume that $z$ is in $P (n_1: n_k)$.
So, any edge in $\cD$ is either above or below $z$ on $P (n_1: n_k)$.
If $l$ is above $z$, then $U(s_1:z)P(z:
t_1)$ is a $(s_1,t_1)$ path which does not contain $l \in \cD$ -
a contradiction. If $l$ is below $z$, then
$T_1P(n_1:z)U(z:t_2)$ is a $(s_2,t_2)$ path
which does not contain $l \in \cD$ - a contradiction.
\end{proof}

\begin{claim}\label{claim:nopath}
(i) There exists a $(s_1,t_2)$ path which 
does not contain any node from $I_{n_1,n_k}$.
(ii) There exists a $(s_2,t_1)$ path which
does not contain any node from $I_{n_1,n_k}$.
\end{claim}

\begin{proof} We prove only the first part. The second part follows
similarly. If the first part is not true, then by the previous claim,
all $(s_1,t_2)$ paths will contain $\cD$. Then removing any
edge in $\cD$ also disconnects $(s_1, t_2)$ - a contradiction. 
\end{proof}

Let $U_1$ and $U_2$ respectively be $(s_1,t_2)$ and $(s_2,t_1)$ paths not
containing any node from
$I_{n_1,n_k}$ (see Fig. \ref{fig:caseIIA-b}).
The network clearly contains a butterfly network as a substructure and
thus has a solution. Specifically, we can transmit $x_1+x_2$ on
$P(n_1,n_k)$, $x_1$ on $U_1$, and $x_2$ on $U_2$. This completes the
proof for Case IIA.

\begin{figure}[h]
\centering
\subfigure[An impossibility in case IIA]{
\includegraphics[width=1.6in]{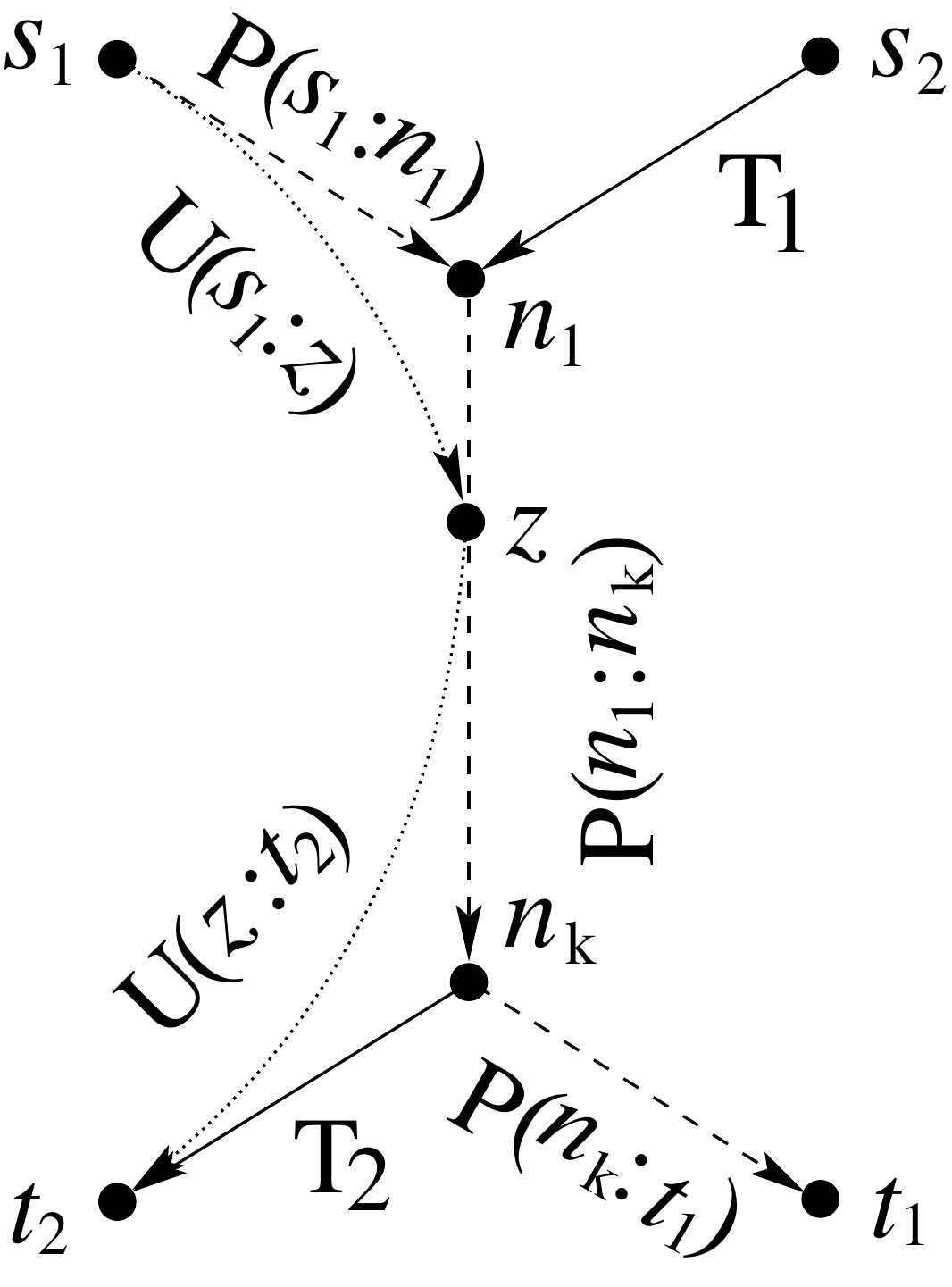}
\label{fig:caseIIA-a}
}
\subfigure[A hidden butterfly]{
\includegraphics[width=1.6in]{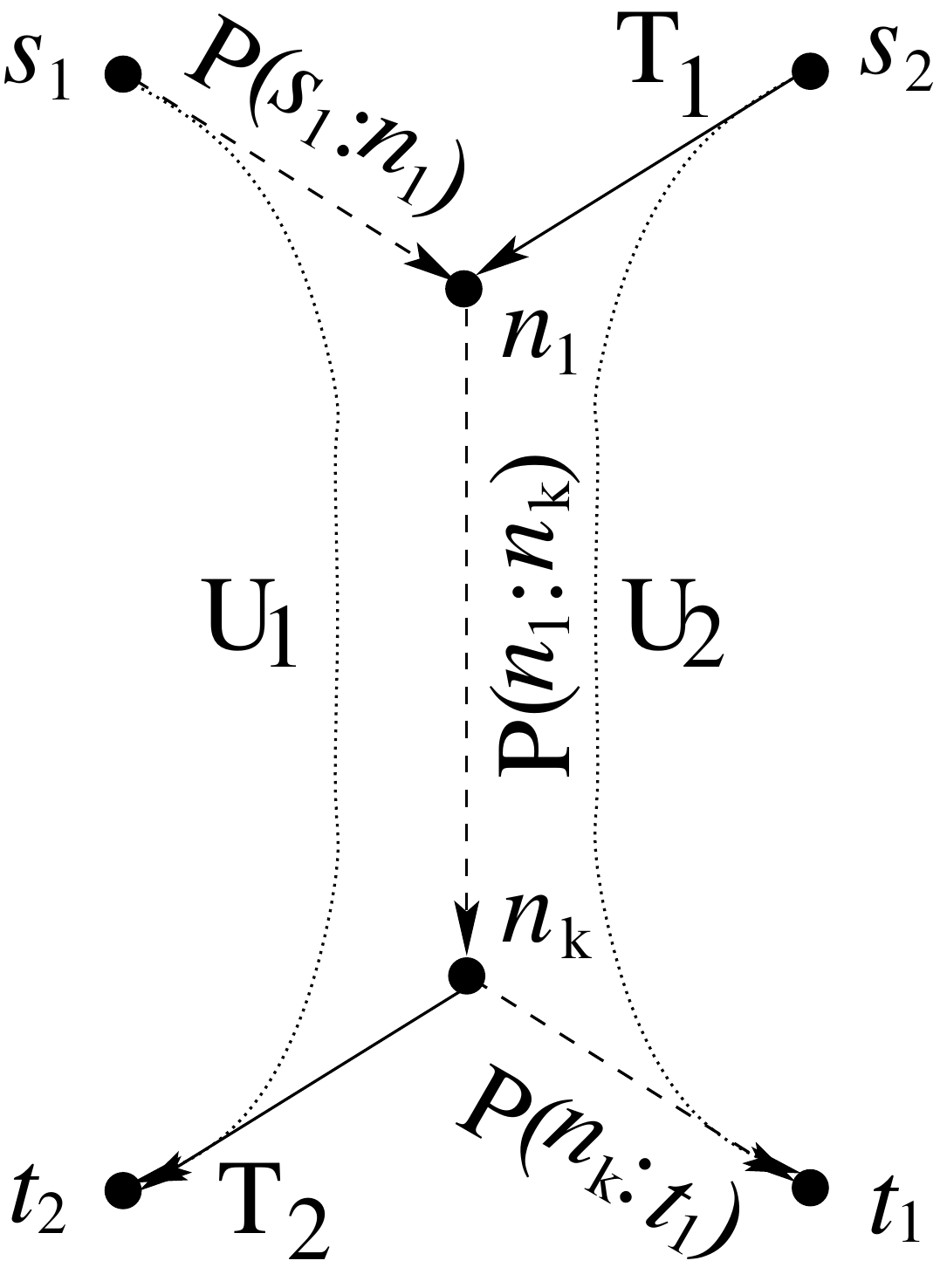}
\label{fig:caseIIA-b}
}
\caption[]{Case IIA}
\label{fig:caseIIA}
\end{figure}

{\bf Case IIB:} {\it There exists no edge in $\cN$ whose removal 
disconnects both $(s_1,t_1)$ and $(s_2,t_2)$.}

Let $\overline{N}=\{n_2,n_3,...,n_{k-1}\}$ be the ordered set of
internal nodes on $P(n_1:n_k)$. Let
$\overline{E}=\{e_1,e_2,...e_{k-1}\}$ be the ordered set of 
edges on $P(n_1:n_k)$
(i.e. $n_i=tail(e_i)$, $i\in\{1,...k-1\}$ and $n_k=head(e_{k-1})$). 
Recall that removing any single edge from $\overline{E}$
 does not disconnect both $(s_1,t_1)$ and $(s_2,t_2)$.

We split our proof into the following three sub-cases:

{\it Case IIB(i):} {\it Removing edge $e_1$ disconnects $(s_2,t_2)$ 
but not $(s_1,t_1)$.}

In this case there is at least one $(s_1,t_1)$ path
which does not contain $e_1$. Consider one such path $P'$.
Note that $P'$ does not contain any node which is on 
any path from $s_2$ to $n_1$ or on any path from
$n_k$ to $t_2$ by definition of $n_1$ and $n_k$.

Let $m(P')$ be the last node on this path so that there 
is no common edge between $P'(s_1:m(P'))$ and 
$\overline{E}$ - that is, $m(P')$ is the tail node of 
the first common edge, if there is such an edge, between $P'$ and
$\overline{E}\backslash\{e_1\}$, or $t_1$ if there is no such common edge.
Thus $m(P')$ will be one of the nodes in 
$\overline{N}\bigcup\{t_1\}$. So, the nodes $m(P^\prime)$ for
all such paths $P^\prime$ are totally ordered based on ancestry. Let us denote 
the last among them by $m_0$ and call the corresponding 
path-section from $s_1$ to $m_0$ by $P_0$. 
So, $P_0$ does not have any edge from $\overline{E}$.

\begin{claim}\label{claim:not}
$m_0\neq t_1$, i.e., $m_0 \in \oN$.
\end{claim}
\begin{proof}
If $m_0=t_1$, then clearly there are two edge disjoint
$(s_1,t_1)$ and $(s_2, t_2)$ paths and so the network would fall in Case I.
\end{proof}

Now, the network may have a variant of the
``grail'' with multiple ``handles'' (See Fig. \ref{fig:multigrail}) or an augmented half-butterfly (See
Fig. \ref{fig:2I-b})
as identified by the following iterative algorithm. In fact, the network
in Fig.~\ref{fig:2I-b} can also be seen as a grail network with an extra
connection.

With the initialization as $i=1$, $v_0 = n_1$, $v_1 = m_0$ and $Q_1 = P_0$,
we find the handles iteratively as follows.

If $v_i = t_1$ or $t_2$, then terminate. Otherwise increment $i$ and define
$v_{i}$ to be the unique last node (ancestrally) such that

(i) there is a path $Q_{i}$ (a new handle) from some node in
$V(P(v_{i-2}: v_{i-1}))\setminus \{v_{i-2}, v_{i-1}\}$ to $v_{i}$ with no
common edge with $\oE$, and

(ii) there is a path from $v_{i}$ to $t_1$ or $t_2$.

\begin{figure}[ht]
\centering
\subfigure[The grail for $I=4$]{
\includegraphics[width=1.6in]{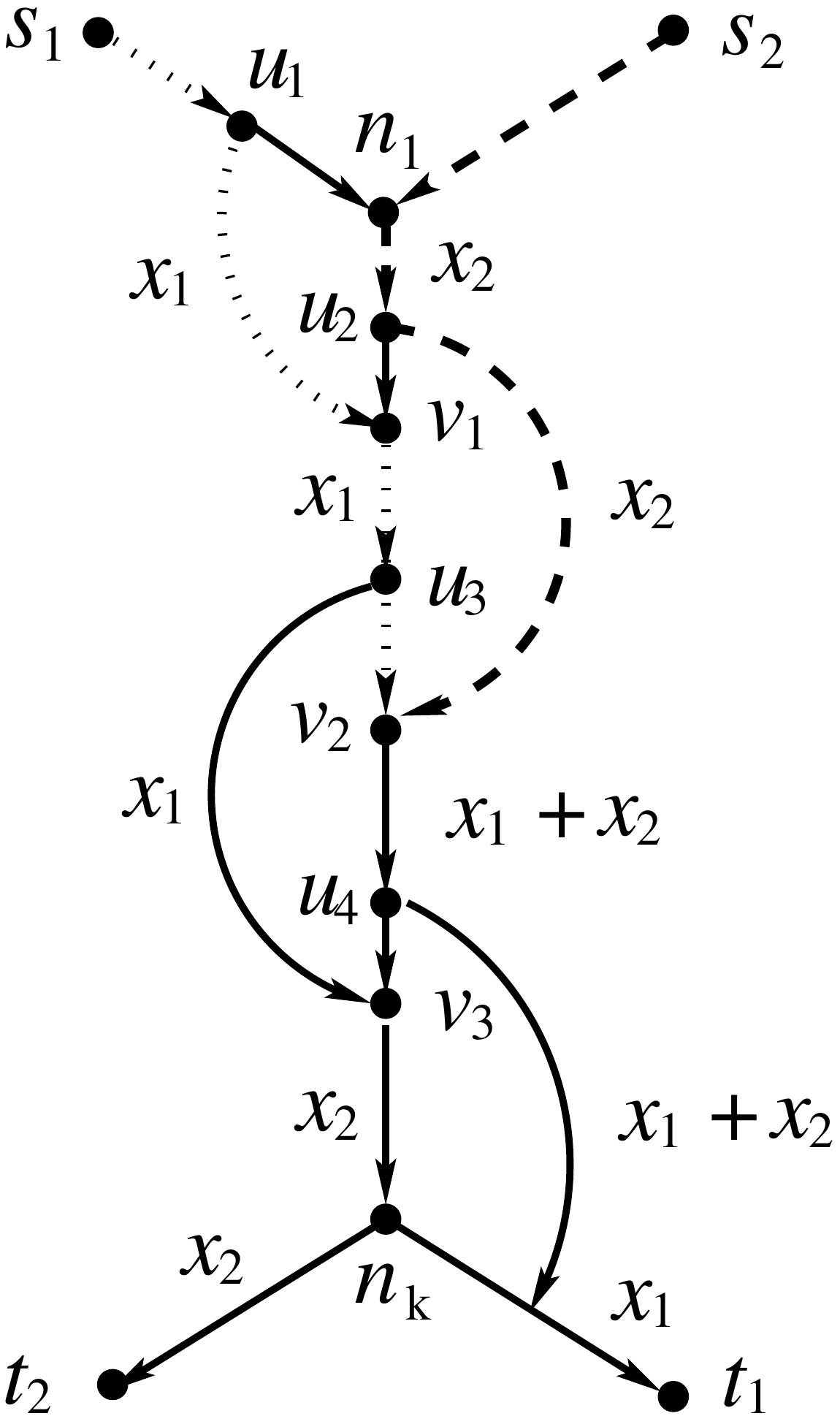}
\label{fig:multigrail-a}
}
\subfigure[The grail for $I=5$]{
\includegraphics[width=1.6in]{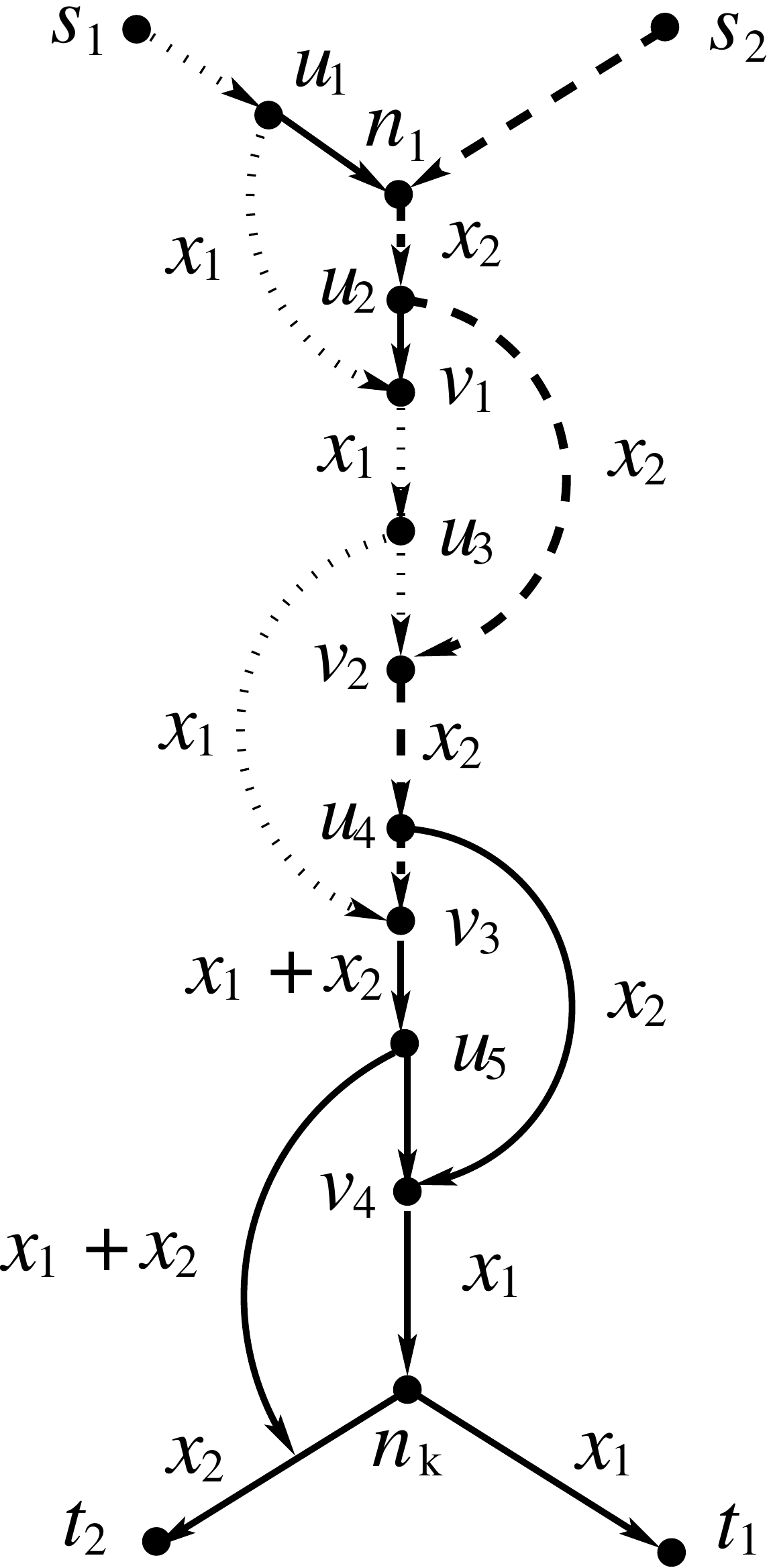}
\label{fig:multigrail-b}
}
\caption[]{The grails with multiple ``handles'' ($I>2$)}
\label{fig:multigrail}
\end{figure}
We denote the first node on $Q_i$ by $u_i$.
We now prove the following claims about this construction.

\begin{claim}\label{claim:exist}
(i) Such a node $v_{i}$ exists.

(ii) Any such $v_{i}$ is in $\oN \cup \{t_1, t_2\}$.

(iii) $v_{i}$ is unique.
\end{claim}
\begin{proof} We outline a proof by induction.

Since $v_1=m_0 \in \oN$, all of the above statements are true
for $i=1$. We prove these statements for $i$, assuming
them to be true for $< i$.
By our induction hypothesis $v_{i-1}$ satisfies all of the  above, particularly $v_{i-1} \in \oN$. Let $e(v_{i-1})$ be the edge in $\oE$ that starts at $v_{i-1}$, i.e., $v_{i-1}$ is the tail of $e(v_{i-1})$. 

By assumption, removing $e(v_{i-1})$ does not disconnect both
$(s_1,t_1)$ and $(s_2,t_2)$. Hence there must be at least one $(s_1,t_1)$
or $(s_2,t_2)$ path which does not contain $e(v_{i-1})$. Let one such
path be $P^*$. If $P^*$ is a $(s_1,t_1)$ path, then it must have a
common edge with $P(n_1:v_{i-1})$ by definition of $m_0$. If $P^*$ is a
$(s_2,t_2)$ path, then it must contain $e_1$ because $e_1$ disconnects
$(s_2, t_2)$. In either case,
by definition of $v_j$ for $j<i$, the last common edge
between $P^*$ and $P(n_1:v_{i-1})$ must have the head node, say $\underline{v}$, in
$V(P(v_{i-2}:v_{i-1}))\setminus\{v_{i-2}, v_{i-1}\}$. Let the first node
from $\{t_1,t_2\} \cup \oN \setminus \{n_1,n_2,...,v_{i-1}\}$ which is on $P^*$
be denoted by $\overline{v}$. Then $P^*(\underline{v}:\overline{v})$ satisfies
the condition of $Q_i$ in the definition of $v_i$. So $\overline{v}$
satisfies both the conditions in the definition of $v_i$, and thus
guarantees the existence of $v_i$. Since $v_i$ is an ancestrally last node
satisfying the two defining conditions, it follows that $v_i \in \oN
\cup \{t_1,t_2\}$. Since this set is totally ordered, $v_i$ is unique.
\end{proof}

Now, if the final iteration is $i=I$, i.e. if $v_{I} = t_1$ or $t_2$,
then the following claims hold. Recall that $Q_i$ is the $i$-$th$ handle (path) from $u_i$ to $v_i$.

\begin{claim}\label{claim:soln}
(i) If $I=2$ and $v_2 = t_1$ then $Q_1$ and $Q_2$ do not share a
common node.

(ii) If $I=2$ and $v_2=t_2$ then $Q_1$ and $Q_2$ share a common edge.

(iii) If $I\neq 2$ then $Q_{i}$ and $Q_j$ do not share a common node
for $1\leq i < j \leq I$. 

(iv) If $I\neq 2$ is even, then there are two edge-disjoint
paths, one from $s_1$ to $v_{I-2}$ via $u_{I-1}$, and the other
from $s_2$ to $v_{I-2}$.

(v) If $I\neq 2$ is odd, then there are two edge-disjoint
paths, one from $s_1$ to $v_{I-2}$, and the other from
$s_2$ to $v_{I-2}$ via $u_{I-1}$.

(vi) If $I\neq 2$ is even, then $v_I=t_1$.

(vii) If $I$ is odd, then $v_I=t_2$.

\end{claim}

\begin{proof} (i) If $Q_1$ $(=P_0)$ and $Q_2$ share a common node $v$, then
$Q_1(s_1:v)Q_2(v:t_1)$ is a $(s_1,t_1)$ path containing
no edge from $\oE$, which means that $m_0=t_1$ - a contradiction to Claim
\ref{claim:not}. 

(ii) If $Q_1$ ($=P_0$) and $Q_2$ are edge-disjoint, then $Q_1P(m_0,t_1)$
and $T_1P(n_1,u_2)Q_2$ are edge-disjoint $(s_1,t_1)$ and $(s_2,t_2)$ paths
- a contradiction.

(iii) If $Q_{i}$ and $Q_j$ share a common node $v$, then $Q_i(u_i:v)Q_j(v:v_j)$
is a path from a node in $V(P(v_{i-2}, v_{i-1}))\setminus
\{v_{i-2}, v_{i-1}\}$ to $v_{j}$ with no common edge with $\oE$.
Then $v_j$, a descendant of $v_{i}$, satisfies all the properties
of $v_{i}$ - a contradiction to the choice of $v_i$.

(iv) The paths $Q_1P(v_1:u_3)Q_3P(v_3:u_5)\cdots$ $Q_{I-3}P(v_{I-3}:u_{I-1}:v_{I-2})$ and $T_1P(n_1:u_2)Q_2P(v_2:u_4)\cdots P(v_{I-4}:u_{I-2})Q_{I-2}$ are two such paths as shown in Fig. \ref{fig:multigrail-a} by dotted and dashed lines respectively for
$I=4$.

(v) The paths are similar to part (iv) and are shown in Fig. \ref{fig:multigrail-b} for $I=5$.

(vi) If $v_I=t_2$, then the paths in the proof of part (iv) together with
the paths $Q_{I-1}P(v_{I-1}:t_1)$ and $P(v_{I-2}:u_I)Q_I$ allow
construction of two edge disjoint $(s_1,t_1)$ and $(s_2,t_2)$ paths
- so the network would fall under Case I.

(vii) Similar to the proof of part (vi).
\end{proof}

So, under Case IIB(i), there are four different types of networks where
the two source-terminal pairs are connected.

1. For $I=2$, $v_2=t_1$, the grail as shown in Fig. \ref{fig:2I-a}.

2. For $I=2$, $v_2=t_2$, the augmented half-butterfly shown in Fig. \ref{fig:2I-b}.

3. For even $I\neq 2$, the tall grail shown in Fig. \ref{fig:multigrail-a} with even
number of handles and $v_I=t_1$.

4. For odd $I\neq 1$, the tall grail shown in Fig. \ref{fig:multigrail-b} with odd
number of handles and $v_I=t_2$.

For all the above cases, simple XOR coding solutions exist as shown in
Fig. \ref{fig:multigrail} and Fig. \ref{fig:2I}.

\begin{figure}[ht]
\centering
\subfigure[The grail]{
\includegraphics[width=1.5in]{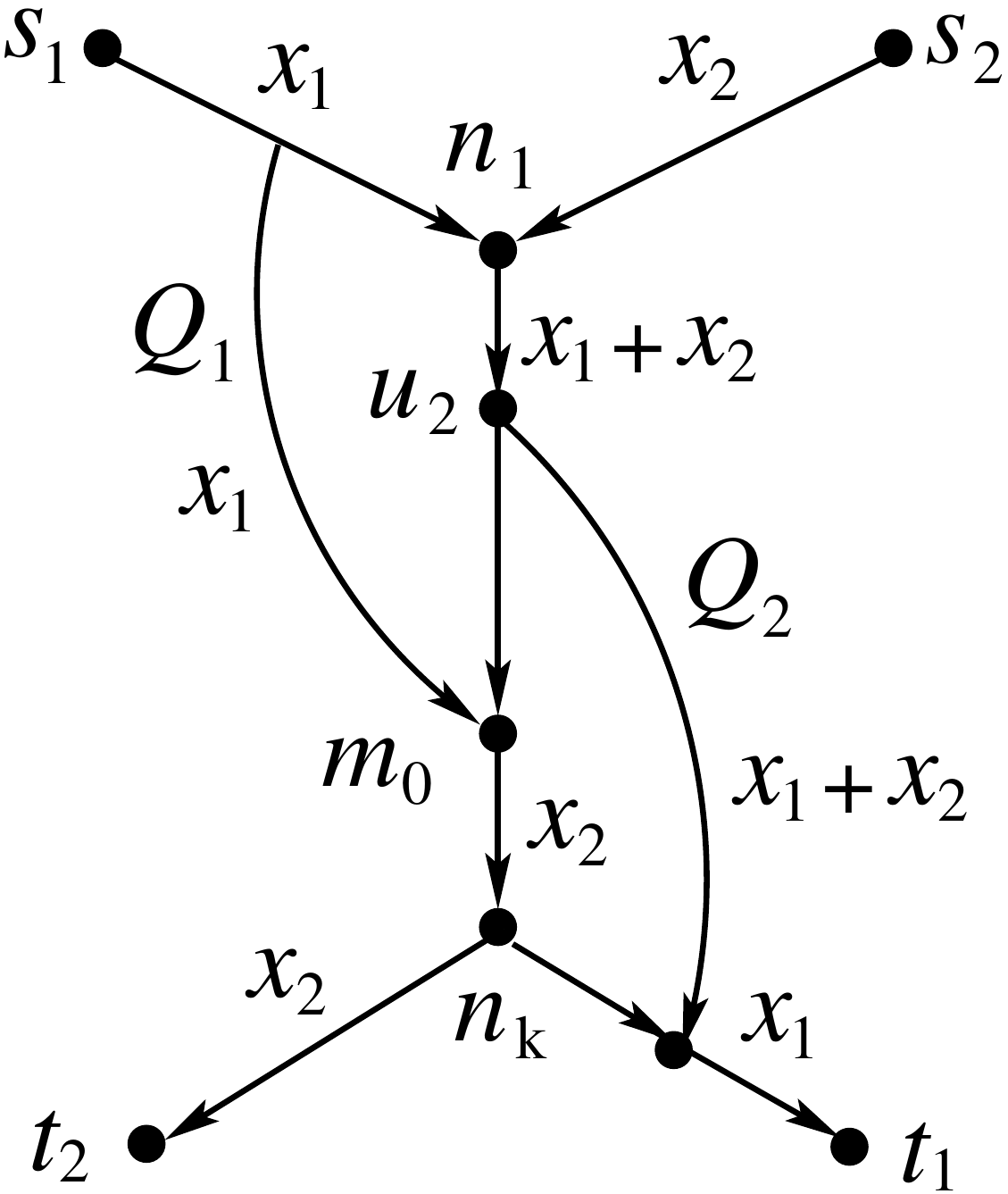}
\label{fig:2I-a}
}
\subfigure[The augmented half-butterfly]{
\includegraphics[width=1.6in]{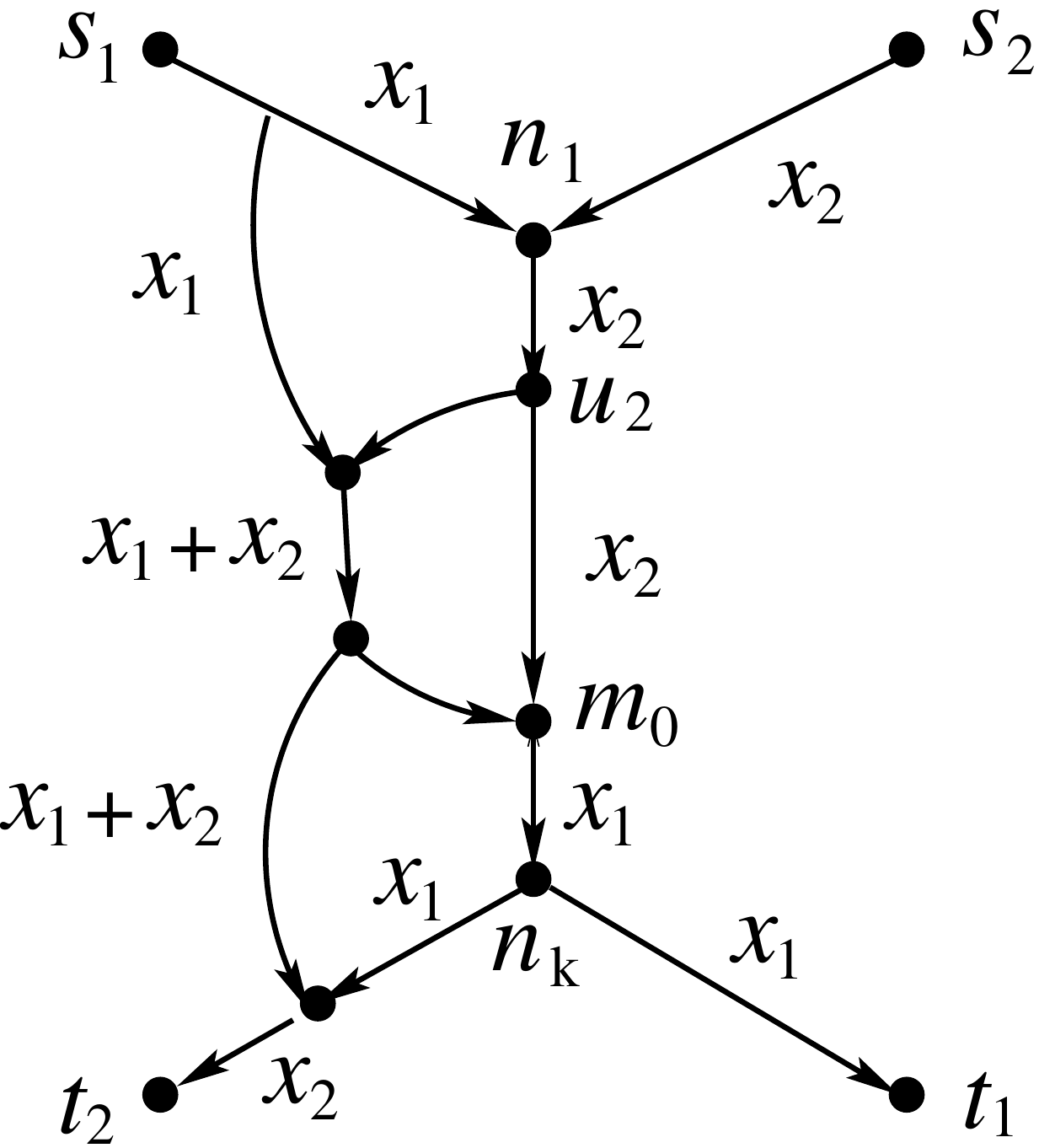}
\label{fig:2I-b}
}
\caption[]{The solvable networks for $I=2$}
\label{fig:2I}
\end{figure}

{\it Case IIB(ii):} {\it Removing edge $e_1$ disconnects $(s_1,t_1)$
but not $(s_2,t_2)$.}

The proof for this case is similar to Case IIB(i).

{\it Case IIB(iii):} {\it Removing edge $e_1$ disconnects neither $(s_1,t_1)$ nor $(s_2,t_2)$.}

In this case, we will show that we can keep removing edges from the network
till the network falls under Case IIB(i) or Case IIB(ii).
Let us denote the $m_0$ defined
in Case IIB(i) as $m_0^{(1)}$, and define $m_0^{(2)}$ similarly considering
$(s_2,t_2)$ paths instead of $(s_1,t_1)$ paths (See Fig. \ref{fig:reduction}). By Claim \ref{claim:not},
$m_0^{(1)}, m_0^{(2)} \in \oN$. Let the corresponding paths
(handles)
$P_0$ be denoted by $P_0^{(1)}$ and $P_0^{(2)}$.
WLOG, let us assume that
$m_0^{(1)}$ is below or same as $m_0^{(2)}$. In this case we will show that it is possible to selectively remove edges till the network falls under Case IIB(i).
Consider any edge on $P(m_0^{(1)}:n_k)$ and denote it by $e$. Also let us
denote the first edge
on $P_0^{(2)}$ which is not on $T_1$ by $e^\prime$. The relevant nodes, edges,
and the paths are shown in Fig. \ref{fig:reduction}. 

\begin{figure}[ht]
\centering
\includegraphics[width=1.3in]{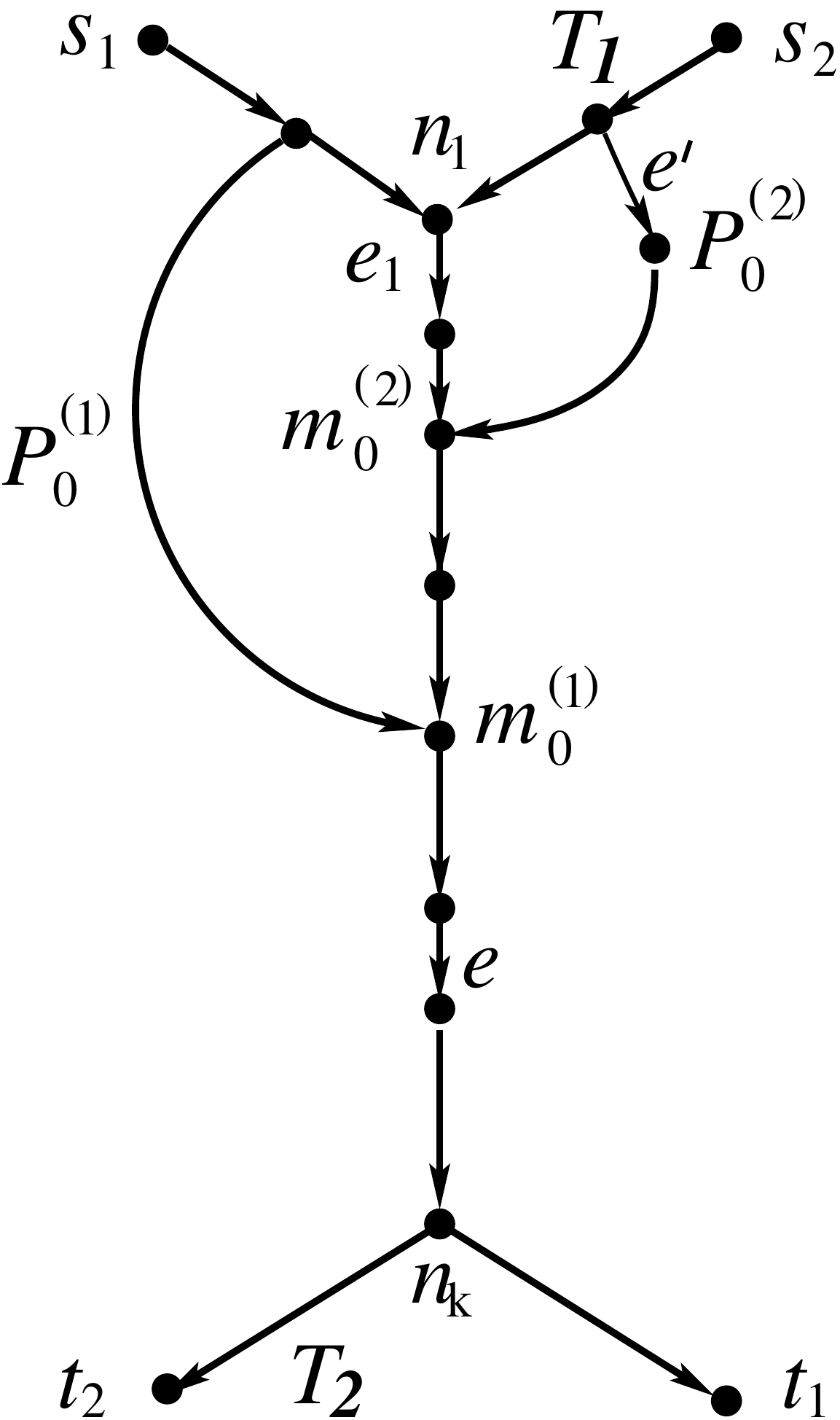}
\caption[]{Reduction of Case IIB(iii) to Case IIB(i)}
\label{fig:reduction}
\end{figure}

\begin{claim}\label{claim:removal}
If removal of $e$ does not disconnect $(s_1,t_1)$ (resp. $(s_2,t_2)$),
then removal of both $e$ and $e^\prime$ also does not disconnect
$(s_1,t_1)$ (resp. $(s_2,t_2)$).
\end{claim}

\begin{proof} By definition of $m_0^{(1)}$ and $m_0^{(2)}$, any
$(s_1,t_1)$ or $(s_2,t_2)$ path $P^{\prime\prime}$
not containing $e$ contains at least one edge, in particular one node $v$,
from the path $P(n_1:tail (e))$. If $P^{\prime\prime}$ is a $(s_1,t_1)$ path
then $P(s_1: v)P^{\prime\prime}(v: t_1)$ is a $(s_1,t_1)$ path not containing
$e$ and $e^\prime$. Similarly, if $P^{\prime\prime}$ is a $(s_2,t_2)$ path
then $T_1P(n_1: v)P^{\prime\prime}(v: t_2)$ is a $(s_2,t_2)$ path not
containing $e$ and $e^\prime$.
\end{proof}

Note that if removal of any edge disconnects $(s_1,t_1)$ or $(s_2,t_2)$, such an edge
must be on $P(n_1:n_k)$. 
So, by Claim \ref{claim:removal}, removing $e^\prime$ will not result in a network
where there is an edge whose removal disconnects both $(s_1,t_1)$
and $(s_2,t_2)$. So, the resulting network will still satisfy Case IIB.
We can thus continue this process of removing selected edges till
the network satisfies Case IIB(i).

\subsection{Proof of sufficiency}
We now prove that if an edge $e$ satisfying the hypothesis
of the theorem exists then the capacity region of the network is given
by $r_1+r_2 \leq 1$. In particular, this implies that the rate $(1,1)$ is
not supported.
WLOG, let us assume that there is an edge $e$ whose removal disconnects
$(s_1,t_1)$, $(s_2,t_2)$ and $(s_2,t_1)$.

Since both the source-terminal pairs are assumed to be connected, there exist
$(s_1,t_1)$ and $(s_2,t_2)$ paths. Then by time-sharing between the achievable
rates $(1,0)$ and $(0,1)$, we can achieve any rate $(r_1,r_2)$ with
$r_1+r_2=1$.
It only remains to be proved that
no rate-pair with $r_1+r_2 > 1$ is achievable.

Note that $e$ may or may not disconnect $(s_1,t_2)$.
Fig. \ref{fig:capacity} depicts the connectivity between the sources and the
terminals and the role of $e$.

Let us assume that the network supports a $n$-length vector (possibly
non-linear) network code by which $\bx_1 \in \cA^{r_1n}$ and
$\bx_2 \in \cA^{r_2n}$ are communicated to $t_1$ and $t_2$ respectively. So, under this coding scheme, each edge carries
an element of $\cA^n$, and the terminals $t_1$ and $t_2$ can recover the source
vectors $\bx_1 \in \cA^{r_1n}$ and $\bx_2\in \cA^{r_2n}$ respectively.
Let $\bY_e = f(\bx_1, \bx_2)  \in \cA^n$ denote the vector carried by the
edge $e$. Since $e$ is a cut between $s_2$ and $t_2$, and since $t_2$
can recover $\bx_2$, $f(\bx_1, \bx_2)$ is a $1-1$ function of $\bx_2$
for any fixed value of $\bx_1$. 

\begin{figure}[ht]
\centering
\includegraphics[width=1.4in]{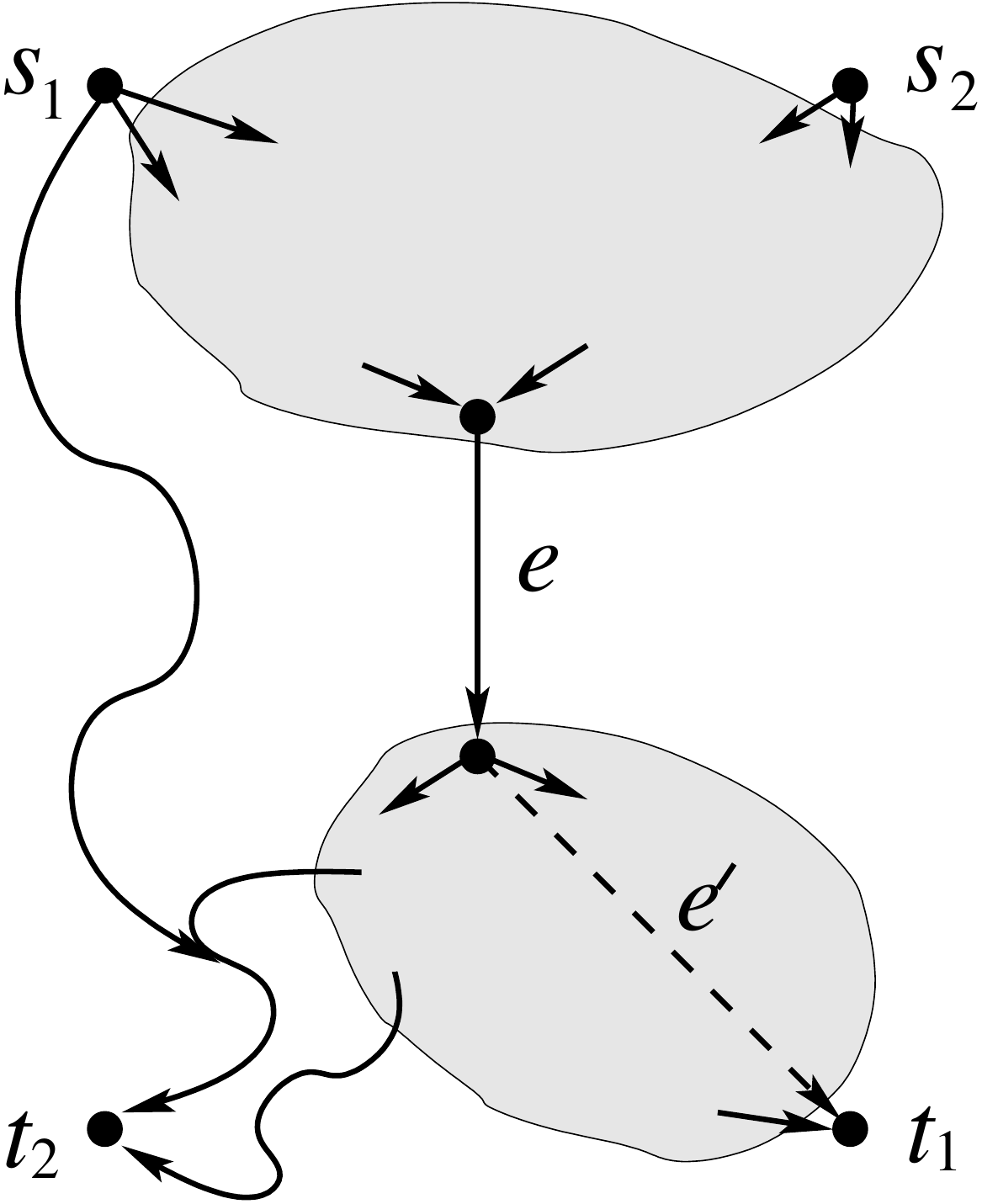}
\caption[]{Connectivity in a network with an edge $e$ satisfying the hypothesis of Theorem~\ref{theorem:mainres}}
\label{fig:capacity}
\end{figure}

\begin{figure}[ht]
\centering
\includegraphics[width=1.6in]{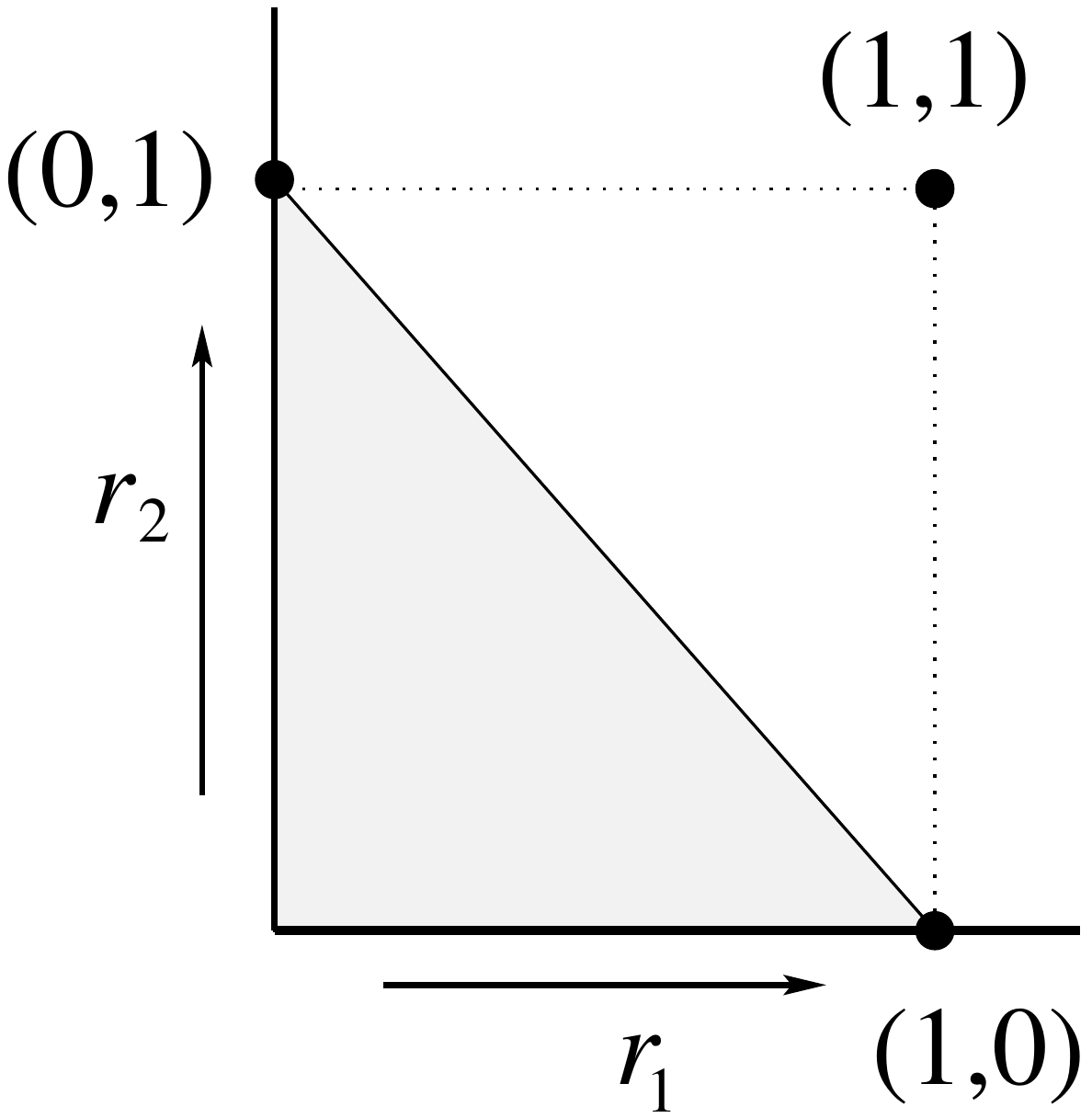}
\caption[]{Capacity region of a network with an edge satisfying the hypothesis
of Theorem~\ref{theorem:mainres}}
\label{fig:capacit}
\end{figure}

Now let us consider a new network obtained by adding an edge
$e^\prime$ from $head(e)$ to $t_1$ (see Fig. \ref{fig:capacity}).
We will show that even for this
stronger network the rate pair $(r_1,r_2)$ must satisfy $r_1+r_2\leq 1$.
Even in this new network, $e$ is a cut
between $\{s_1,s_2\}$ and $t_1$. Since $t_1$ recovers $\bx_1$ and
$f(\bx_1,\bx_2)$ is a $1-1$ function of $\bx_2$ for a given $\bx_1$,
$t_1$ can recover $\bx_2$ as well in this new network. Hence in the new network
$t_1$ recovers both $\bx_1$ and $\bx_2$. So, $f$ is a $1-1$ function
of the pair $(\bx_1,\bx_2)$. So, $|\cA |^n \geq |\cA|^{r_1n +r_2n}
\Rightarrow r_1+r_2 \leq 1$.
\section{Conclusion}\label{section:conclusion}
In this paper, we gave a very simple necessary and sufficient condition for
a directed acyclic network to simultaneously support two unicast sessions at
rate $1$ under vector network coding, and used
this to provide an exact capacity region characterization,
namely $r_1+r_2 \leq 1$,
of networks which do not support the rate $(1,1)$.

\section{Acknowledgment}
The authors thank Chih-Chun Wang for pointing out the equivalence of
our condition with the network sharing bound, and the grail in Fig.~\ref{fig:2I-b}.
This work was supported in part by Bharti Centre for Communication at
Indian Institute of Technology Bombay, and a grant from Department of
Science and Technology (DST), India.

\end{document}